\documentclass[10pt,a4paper,reqno]{amsart}
\usepackage[utf8]{inputenc}
\usepackage[T1]{fontenc}
\usepackage{lmodern}
\usepackage{asymptote}
\usepackage{common}

\numberwithin{equation}{section}
\numberwithin{figure}{section}
\numberwithin{table}{section}
\numberwithin{theorem}{section}
\numberwithin{lemma}{section}

\newcommand{\loja}{\L{}ojasiewicz\xspace}

\begin{document}

\title{Convergence of gradient-based algorithms for the Hartree-Fock
  equations}
\date{\today}

\address{Université Paris-Dauphine, CEREMADE, Place du Maréchal Lattre
  de Tassigny, 75775 Paris Cedex 16, France.}
\email{levitt@ceremade.dauphine.fr} \thanks{Support from the grant
  ANR-10-BLAN-0101 of the French Ministry of Research is gratefully
  acknowledged}


\begin{abstract}
  The numerical solution of the Hartree-Fock equations is a central
  problem in quantum chemistry for which numerous algorithms
  exist. Attempts to justify these algorithms mathematically have been
  made, notably in \cite{cances2000}, but, to our knowledge, no
  complete convergence proof has been published. In this paper, we
  prove the convergence of a natural gradient algorithm, using a
  gradient inequality for analytic functionals due to
  \loja\cite{lojasiewicz1965ensembles}. Then, expanding upon the
  analysis of \cite{cances2000}, we prove convergence results for the
  Roothaan and Level-Shifting algorithms. In each case, our method of
  proof provides estimates on the convergence rate. We compare these
  with numerical results for the algorithms studied.
\end{abstract}

\maketitle
\section{Introduction}
In quantum chemistry, the Hartree-Fock method is one of the simplest
approximations of the electronic structure of a molecule. By assuming
minimal correlation between the $N$ electrons, it reduces
Schr\"odinger's equation, a linear partial differential equation on
$\R^{3N}$, to the Hartree-Fock equations, a system of $N$ coupled
nonlinear equations on $\R^{3}$. This approximation makes it much more
tractable numerically. It is used both as a standalone description of
the molecule and as a starting point for more advanced methods, such
as the Møller-Plesset perturbation theory, or multi-configuration
methods. Mathematically, the Hartree-Fock method leads to a coupled
system of nonlinear integro-differential equations, which are
discretized by expanding the solution on a finite Galerkin basis. The
resulting nonlinear algebraic equations are then solved iteratively,
using a variety of algorithms, the convergence of which is the subject
of this work.

The mathematical structure of the Hartree-Fock equations was
investigated in the 70's, culminating in the proof of the existence of
solutions by Lieb and Simon\cite{lieb1977}, later generalized by
Lions\cite{lions1987}. On the other hand, despite their ubiquitous use
in computational chemistry, the convergence of the various algorithms
used to solve them is still poorly understood. A major step forward in
this direction is the recent work of Cancès and Le
Bris\cite{cances2000}. Using the density matrix formulation, they
provided a mathematical explanation for the oscillatory behavior
observed in the simplest algorithm, the Roothaan method, and proposed
the Optimal Damping Algorithm (ODA), a new algorithm inspired directly
by the mathematical structure of the constraint
set\cite{cances2000can}. This algorithm was designed to decrease the
energy at each step, and linking the energy decrease to the difference
of iterates allowed the authors to prove that this algorithm
``numerically converges'' in the weak sense that $\norm{D_k - D_{k-1}}
\to 0$. In addition, the algorithm numerically converges towards an
Aufbau solution\cite{cancesODA}. This, to our knowledge, is the
strongest convergence result available for an algorithm to solve the
Hartree-Fock equations.

However, this is still mathematically unsatisfactory, as it does not
guarantee convergence, and merely prohibits fast divergence. The
difficulty in proving convergence of the algorithms used to solve the
Hartree-Fock equations lies in the lack of understanding of the
second-order properties of the energy functional (for instance, there
are no local uniqueness results available). In other domains, the
convergence of gradient-based methods has been established using the
\loja inequality for analytic
functionals\cite{lojasiewicz1965ensembles} (see for instance
\cite{salomon2005convergence, haraux2003rate}). This method of proof
has the advantage of not requiring any second-order information.

In this paper, we use a gradient descent algorithm to solve the
Hartree-Fock equations. This algorithm builds upon ideas from
differential geometry\cite{edelman1998geometry} and the various
projected gradient algorithms used in the context of quantum
chemistry\cite{cances2008projected,mcweeny1956density,alouges2009preconditioned}. To
our knowledge, this particular algorithm has never been applied to the
Hartree-Fock equations. Although it lacks the sophistication of modern
minimization methods (for instance, see the trust region methods of
\cite{francisco2004globally} and \cite{host2008augmented}), it is the
most natural generalization of the classical gradient descent, and, as
such, the simplest one to analyze mathematically. For this algorithm,
following the method of \cite{salomon2005convergence}, we prove
convergence, and obtain explicit estimates on the convergence rate. We
also apply the method to the widely used Roothaan and Level-Shifting
algorithms, effectively linking these fixed-point algorithms to
gradient methods.

The structure of this paper is as follows. We first introduce the
Hartree-Fock problem in the mathematical setting of density matrices
and prove a \loja inequality on the constrained parameter space. We
then introduce the gradient algorithm, and prove some estimates. We
show the convergence and obtain convergence rates for this algorithm,
then extend our method to the Roothaan and Level-Shifting algorithm,
using an auxiliary energy functional following \cite{cances2000}. We
finally test all these results numerically and compare the convergence
of the algorithms.
\section{Setting}
\label{setting}

We are concerned with the numerical solution of the Hartree-Fock
equations. We will consider for simplicity of notation the spinless
Hartree-Fock equations, where each orbital $\phi_i$ is a function in
$L^2(\R^3,\R)$, although our results are easily transposed to other
variants such as General Hartree-Fock (GHF) and Restricted
Hartree-Fock (RHF).

In this paper, we consider a Galerkin discretization space with finite
orthonormal basis $(\chi_i)_{i=1\dots N_b}$. In this setting, the
orbitals $\phi_i$ are expanded on this basis, and the operators we
consider are $N_b \times N_b$ matrices. This finite dimension
hypothesis is crucial for our results, and we comment on it in the
conclusion.

The Hartree-Fock problem consists in minimizing the total energy of a
N-body system. We describe the mathematical structure of the energy
functional and the minimization set, and introduce a natural gradient
descent to solve this problem numerically.
\subsection{The energy}
We consider the quantum N-body problem of $N$ electrons in a potential
$V$ (in most applications, $V$ is the Coulombic potential created by a
molecule or atom). In the spinless Hartree-Fock model, this problem is
simplified by assuming that the N-body wavefunction $\psi$ is a single
Slater determinant of $N$ $L^{2}$-orthonormal orbitals $\phi_i$. A
simple calculation then shows that the energy of the wavefunction
$\psi$ can be expressed as a function of the orbitals $\phi_i$,
\begin{align*}
  \mathcal E(\psi) = \sum_{i=1}^N \int_{\R^3} \frac 1 2 (\nabla \phi_i)^2 +
  \int_{\R^3} V \rho + \frac 1 2
  \int_{\R^3} \int_{\R^3} \frac{\rho(x) \rho(y) - \tau(x,y)^2}{|x-y|}
  \dif x \dif y,
\end{align*}
where $\tau(x,y) = \sum_{i=1}^N \phi_i(x) \phi_i(y)$ and $\rho(x) =
\tau(x,x)$.

The energy is then to be minimized over all sets of orthonormal
orbitals $\phi_i$. An alternative way of looking at this problem is to
reformulate it using the density operator $D$. This operator, defined
by its integral kernel $D(x,y) = \tau(x,y)$, can be seen to be the
orthogonal projection on the space spanned by the $\phi_i$'s. The energy
can be written as a function of $D$ only:
\begin{align}
  E(D) = \Tr((h + \frac 1 2 G(D)) D),
\end{align}
where
\begin{align*}
  h &= -\frac 1 2 \Delta + V,\\
  (G(D) \phi)(x) &= \lp\rho \star \frac 1 {|\cdot|}\rp(x) \phi(x) -
  \int_{y} \frac{\phi(y) \tau(x,y)}{|x-y|} .
\end{align*}
This time, the energy is to be minimized over all orthogonal projection
operators of rank $N$. In the discrete setting, the orbitals
$\phi_{j}$ are discretized as $\phi_{j} = \sum_{i=1}^{N_{b}} c_{ij}
\chi_{i}$, and the operators $D$, $h$, and $G(D)$ become $N_{b} \times
N_{b}$ matrices.

\subsection{The manifold $\mathcal P$}
\label{manifold}
The Hartree-Fock energy is to be minimized over the set of density
matrices
\begin{align*}
  D \in \mathcal P = \{ D \in M_{N_{b}}(\R), D^T = D, D^2 = D, \Tr D = N\}.
\end{align*}

The manifold $\mathcal P$ is equipped with the canonical inner product
in $M_{{N_{b}}}(\R)$
\begin{align*}
  \lela A,B \rira &= \Tr (A^T B).
\end{align*}
We denote by $\norm{A} = \sqrt{\lela A, A\rira}$ the Frobenius (or
Hilbert-Schmidt) norm of $A$, which is the most natural here, and by
$\norm{A}_\text{op} = \max_{\norm{x} = 1} \norm{A x}$ the operator
norm of $A$.

The Riemannian structure of $\mathcal P$ allows us to compute the
gradient of $E$. The tangent space $T_D \mathcal P$ at a point $D$ is
the set of $\Delta$ such that $D + \Delta$ verifies the constraints up
to first order in $\Delta$, that is, such that $\Delta^T = \Delta, D
\Delta + \Delta D = \Delta, \Tr \Delta = 0$. Block-decomposing
$\Delta$ on the two orthogonal spaces $\text{range} (D)$ and $\text{ker} (D)$,
this implies that the tangent space $T_D \mathcal P$ is the set of
matrices $\Delta$ of the form
\begin{align*}
  \Delta &= \mat{0&A^T\\A&0},
\end{align*}
where $A$ is an arbitrary $(N_{b}-N) \times N$ matrix.

Hence, the projection on the tangent space of an arbitrary symmetric
matrix $M$ is given by
\begin{align*}
  P_D(M) &= D M (1-D) + (1-D) M D\\
  &= [D,[D,M]].
\end{align*}

Note that if $M$ has decomposition $\mat{B&A^{T}\\A&C}$, then
$[D,M] = \mat{0&A^{T}\\-A&0}$ and $[D,[D,M]] = \mat{0&A^{T}\\A&0}$,
which shows that $\norm{[D[D,M]]} = \norm{[D,M]}$, a property that
will be useful in the sequel.

We can now compute the gradient of $E$. First, the unconstrained
gradient in $M_{N_{b}}(\R)$ is
\begin{align*}
  \nabla E(D) &= F_D = h + G(D),
\end{align*}
the Fock operator describing the mean field generated by the electrons
of $D$. We obtain the constrained gradient $\nabla_{\mathcal P} E$ by
projecting onto the tangent space:
\begin{align*}
  \nabla_{\mathcal P} E(D) &= P_D(\nabla E(D))\\
  &= [D,[D,F_D]].
\end{align*}

\subsection{\loja inequality}
The \loja inequality for a functional $f$ around a critical point
$x_0$ is a local inequality that provides a lower bound on
$\nabla f$. Its only hypothesis is analyticity. In particular, no
second order information is needed, and the inequality accommodates
degenerate critical points.
\subsubsection{The classical \loja inequality}
\begin{theorem}[\loja inequality in $\R^n$]
  Let $f$ be an analytic functional from $\R^n$ to $\R$. Then, for
  each $x_0 \in \R^n$, there is a neighborhood $U$ of $x_0$ and two
  constants $\kappa > 0$, $\theta \in (0,1/2]$ such that when $x \in
  U$,
  \begin{align*}
    \abs{f(x) - f(x_0)}^{1-\theta} \leq \kappa \norm{\nabla f (x)}.
  \end{align*}
\end{theorem}

This inequality is trivial when $x_0$ is not a critical point. When
$x_{0}$ is a critical point, a simple Taylor expansion shows that, if
the Hessian $\nabla^2 f(x_0)$ is invertible, we can choose $\theta =
\frac 1 2$ and $\kappa > \frac1 {\sqrt{ 2 |\lambda_1|}}$, where
$\lambda_1$ is the eigenvalue of lowest magnitude $\nabla^2
f(x_0)$. When $\nabla^2 f(x_0)$ is singular (meaning that $x_0$ is a
degenerate critical point), the analyticity hypothesis ensures that
the derivatives cannot all vanish simultaneously, and that there
exists a differentiation order $n$ such that the inequality holds with
$\theta = \frac 1 n$.

\subsubsection{\loja inequality on $\mathcal P$}
We now extend this inequality to functionals defined on the manifold
$\mathcal P$.

\begin{theorem}[\loja inequality on $\mathcal P$]
  \label{lojap}
  Let $f$ be an analytic functional from $\mathcal P$ to $\R$. Then, for
  each $D_0 \in \mathcal P$, there is a neighborhood $U$ of $D_0$ and two
  constants $\kappa > 0$, $\theta \in (0,1/2]$ such that when $D \in
  U$,
  \begin{align*}
    \abs{f(D) - f(D_0)}^{1-\theta} \leq \kappa \norm{\nabla_{\mathcal P} f (D)}.
  \end{align*}
\end{theorem}
\begin{proof}
  Let $D_0 \in \mathcal P$. Define the map $R_{D_0}$ from $T_{D_0} \mathcal P$
  to $\mathcal P$ by
  \begin{align*}
    R_{D_0}(\Delta) &= U D_0 U^T,\\
    U &= \exp(-[D_0,\Delta]).
  \end{align*}

  \begin{figure}[H]
    \centering
    \scalebox{0.7}{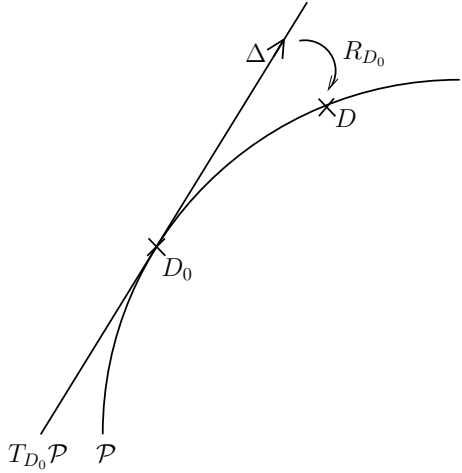}
    \caption{The map $R_{D_0}$}
    \label{fig:manifolds}
  \end{figure}

  This map is analytic and verifies $R_{D_0}(0) = D_0$, $\dif
  R_{D_0}(0) = \text{id}_{T_{D_0} \mathcal P}$. Therefore, by the
  inverse function theorem, the image of a neighborhood of zero
  contains a neighborhood of $D_0$. We now compute the gradient of
  $\tilde f = f \circ R_{D_0}$ at a point $\Delta$, with $D =
  R_{D_0}(\Delta)$.
  \begin{align*}
    \tilde f(\Delta + \delta) &= f(D) + \lela \nabla_{\mathcal P} f(D), \dif
    R_{D_0}(\Delta)  \delta\rira + O(\delta^2)\\
    &= f(D) + \lela \dif R_{D_0}(\Delta)^* \nabla_{\mathcal P} f(D),
    \delta\rira + O(\delta^2)\\
    &= f(D) + \lela P_{D_0} \dif R_{D_0}(\Delta)^* \nabla_{\mathcal P} f(D),
    \delta\rira + O(\delta^2).
  \end{align*}

  We deduce
  \begin{align*}
    \nabla_{T_{D_0} \mathcal P} \tilde f(\Delta) &= P_{D_0} \dif R_{D_0}(\Delta)^* \nabla_{\mathcal P} f(D).
  \end{align*}

  We can now apply the \loja inequality to $\tilde f$, which is an
  analytic functional defined on the Euclidean space $T_{D_0} \mathcal
  M$. We obtain a neighborhood of zero in $T_{D_0} \mathcal P$, and
  therefore a neighborhood $U$ of $D_0$ on which
  \begin{align*}
    \abs{f(D) - f(D_0)}^{1-\theta} \leq \kappa \norm{P_{D_0} \dif R_{D_0}(\Delta)^* \nabla_{\mathcal P} f(D)}.
  \end{align*}
  As $\dif R_{D_0}^*$ is continuous in $\Delta$, up to a change
  in the neighborhood $U$ and the constant $\kappa$,
  \begin{align*}
    \abs{f(D) - f(D_0)}^{1-\theta} \leq \kappa \norm{\nabla_{\mathcal P} f(D)}.
  \end{align*}
\end{proof}

\section{The gradient method}
\subsection{Description of the method}
The gradient flow
\begin{align}
  \od D t &= - \nabla_{\mathcal P} E(D)\notag\\
  &= - [D,[D,F_D]]\label{gradientflow}
\end{align}
is a way of minimizing the energy over the manifold $\mathcal P$. This
continuous flow was already used to solve the Hartree-Fock equations
in \cite{alouges2009preconditioned} (although the authors used a
formulation in terms of orbitals, whereas we use density matrices).

The naive Euler discretization
\begin{align*}
  D_{k+1} &= D_k - t [D_k,[D_k,F_k]]
\end{align*}
is not suitable because it does not stay on $\mathcal P$. A variety of
approaches deal with this problem. One of the first algorithms to
solve the Hartree-Fock equations \cite{mcweeny1956density} used a
purification method to project $D_{k+1}$ back onto $\mathcal P$. More
recently, an orthogonal projection on the convex hull of $\mathcal P$
was used for that purpose \cite{cances2008projected}. Although we
focus in this paper on a different gradient method, such projection
methods have the same behavior for small stepsizes and can be treated
in the same framework, provided one can prove results similar to
Lemmas \ref{lemma_control_curv} and \ref{lemma_first_order} below.

We look for $D_{k+1}$ on a curve on $\mathcal P$ that is tangent to
$\nabla_{\mathcal P} E(D_k)$. A natural curve on $\mathcal P$ is the
change of basis
\begin{align*}
  D'(t) &= U_t D U_t^T,
\end{align*}
where $U_t$ is a smooth family of orthogonal matrices. If we take
\begin{align*}
  U_t &= \exp(t [D,F_D]),
\end{align*}
we get
\begin{align*}
  \eval{\od {D'} t}_{t=0} &= - [D,[D,F_{D}]],
\end{align*}
so $D'(t)$ is a smooth curve on $\mathcal P$, tangent to the gradient
flow at $t = 0$.


Our gradient method with a fixed step $t$ is then
\begin{align}
  D_{k+1} = U_k D_k U_{k}^{T},
\end{align}
with
\begin{align}
  U_k &= \exp(t [D_k,F_k]).
\end{align}
This method, as well as various generalizations, is described in
\cite{edelman1998geometry}.

We now prove a number of lemmas which are the main ingredients of the
convergence proof. First, we bound the second derivative of the energy
to obtain quantitative estimates on the energy decrease, then we link
the difference of iterates $D_{k+1} - D_k$ to the gradient
$\nabla_{\mathcal P}E(D_k)$, and finally we use the \loja inequality
to establish convergence.
\subsection{Derivatives}
We start from a point $D_0$ and compute the derivatives of $E$ along
the curve $D_t = U_t D_0 U_{-t}$. For ease of notation we will write
$\epsilon(t) = E(D_{t})$, $F_t = F(D_t)$ and $C_t = [D_t, F_t]$.
\begin{align*}
  \od {D_{t}} t &= \od {U_t} t D_0 U_{-t} + U_t D_0 \od {U_{-t}} t\\
  &= [C_0,D_t],\\
  \od[n] {D_{t}} t &= \od[n-1] {}t [C_0, D_t]\\
  &= \underbrace{[C_0,[C_0,\dots[C_0,D_t]\dots]]}_{n \text{ times }C_0},\\
  \od \epsilon t  &= \Tr(F_t [C_0,D_t]),\\
  \eval{\od \epsilon t}_{t=0}&= - \norm{C_0}^2,\\
  \od[2] \epsilon t &= \Tr(F_t [C_0,[C_0,D_t]]) + \Tr (G([C_0,D_t])[C_0,D_t]).
\end{align*}

\subsection{Control on the curvature}
\begin{lemma} \label{lemma_control_curv}There exists $\alpha > 0$ such that for every $D_0$ and
  $t$,
  \begin{align*}
    \abs{\od[2] \epsilon t} (t) \leq \alpha \norm{C_0}^2.
  \end{align*}
\end{lemma}
\begin{proof}
  \begin{align}
    \label{curv}
    \od[2] \epsilon t &= \Tr(F_t [C_0,[C_0,D_t]]) + \Tr (G([C_0,D_t])[C_0,D_t]).
  \end{align}

  First, since the Laplacian in $F(D)$ acts on a finite dimensional
  space, we can bound $F(D)$:
  \begin{align}
    \norm{F(D)}_{\text{op}} &\leq \frac 1 2 \norm{-\Delta}_{\text{op}}
    + \norm{V}_{\text{op}} + \norm{G(D)}_{\text{op}}\notag\\
    \label{bound_F}&\leq \frac 1 2 \norm{-\Delta}_{\text{op}} + 2(2N+Z) \sqrt{\norm{-
        \Delta}_\text{op}}
  \end{align}
  by the Hardy inequality. Next, making use of the operator inequality
  $\Tr(AB) \leq \norm{A}_{\text{op}} \norm{B}$, we show that
  \begin{align*}
    \Tr(F_t [C_0,[C_0,D_t]]) &\leq 2 \lp \frac 1 2
    \norm{-\Delta}_{\text{op}} + 2(2N+Z) \sqrt{\norm{-
        \Delta}_\text{op}}\rp \norm{C_{0}}^{2}.
  \end{align*}

  For the second term of \eqref{curv},
  \begin{align*}
    \Tr (G([C_0,D_t])[C_0,D_t])&\leq \norm{G([C_0,D_t])}_\text{op} \Tr\lp\abs{[C_0,D_t]}\rp\\
    &\leq 4  \sqrt {\norm{-\Delta}_\text{op}}
    \Tr\lp\abs{[C_0,D_t]}\rp^2\\
    &\leq 16 N \sqrt {\norm{-\Delta}_\text{op}} \norm{C_{0}}^{2}.
  \end{align*}

  The result is now proved with
  \begin{align*}
    \alpha = \norm{-\Delta}_\text{op} + 4(6N+Z) \sqrt{\norm{- \Delta}_\text{op}}.
  \end{align*}

\end{proof}
\subsection{Choice of the stepsize}

We can now expand the energy:
\begin{align*}
  \epsilon(t) \leq \epsilon(0) - t \norm{C_0}^2 + \frac {t^2} 2 \alpha \norm{C_0}^2.
\end{align*}

If we choose
\begin{align}
  t &< \frac 2 {\alpha},\label{stepsize_limit}
\end{align}
we obtain a decrease of the energy
\begin{align}
  \epsilon(t) &\leq \epsilon(0) - \beta \norm{C_0}^2\label{energy_decrease}
\end{align}
with $\beta = t - \frac  {t^{2}}2 \alpha > 0$.

The optimal choice of $t$ with this bound on the curvature is $t =
\frac 1 \alpha$, with $\beta = \frac 1 {2 \alpha}$. Of course it could
be that the actual optimal $t$ is different, and could vary wildly,
which is why we will not consider optimal stepsizes.

\subsection{Study of $D_{k+1} - D_k$}
We now prove that our iteration $D_{k+1} = U_{k} D_{k} U_{k}^{T}$
coincides with an unconstrained gradient method up to first order in
$t$.

We say that $M \in o(\norm N)$ when for all $\eps > 0$, there is a
neighborhood $U$ of zero such that when $N \in U$, $\norm{M} \leq \eps
\norm{N}$. Note that this neighborhood $U$ is not allowed to depend on
$N$, meaning that the resulting bound is uniform, which will allow us
to manipulate the remainders more easily.
\begin{lemma}
  \label{lemma_first_order}
  For any $k$,
  \begin{align*}
    D_{k+1} &= D_{k} + t [C_{k},D_{k}] + o(t \norm{C_{k}}).
  \end{align*}
\end{lemma}
\begin{proof}
  \begin{align*}
    D_{k+1} - D_k - t [C_k,D_k] &= \sum_{n=2}^\infty \frac {t^n}{n!}
    \underbrace{[C_k,[C_k,\dots[C_k, D_k]\dots]]}_{n \text{ times } C_k}\\
    \norm{D_{k+1} - D_k - t [C_k,D_k]}&\leq t
    \norm{[C_k,D_k]}\sum_{n=2}^\infty t^{n-1} \norm{C_k}^{n-1}\\
    &\leq t \norm{[C_k,D_k]}\frac{t \norm{C_k}}{1 - t \norm{C_k}}
  \end{align*}
\end{proof}

\section{Convergence of the gradient algorithm}
\label{conv}
\begin{theorem}[Convergence of the gradient algorithm]
  \label{thm_cv}
  Let $D_0 \in \mathcal P$ be any density matrix and $D_k$ be the
  sequences of iterates generated from $D_0$ by $D_{k+1} = U_k D_k
  U_{k}^{T}$, with stepsize $t < \frac 2 \alpha$. Then $D_k$ converges
  towards a solution of the Hartree-Fock equations.
\end{theorem}
\begin{proof}
  The energy $E(D)$ is bounded from below on $\mathcal P$, and
  therefore $E_k$ converges to a limit $E_\infty$. In the sequel we
  will work for convenience with $\tilde E(D) = E(D) - E_\infty$ and
  drop the tildes. Immediately, summing \eqref{energy_decrease}
  implies that $C_k$ converges to 0, and therefore so does $D_k -
  D_{k-1}$ (this is what Cancès and Le Bris call ``numerical
  convergence'' in \cite{cances2000}). Note that we only get that
  $\norm{D_k - D_{k-1}}^2$ is summable, which alone is not enough to
  guarantee convergence (the harmonic series $x_k = \sum_{l=1}^k 1/l$
  is a simple counter-example). To obtain convergence, we shall use
  the \loja inequality.

  Let us denote by $\Gamma$ the level set $\Gamma =\{D \in \mathcal P,
  E(D) = 0\}$. It is non-empty and compact. We apply the \loja
  inequality to every point of $\Gamma$ to obtain a cover $(U_i)_{i
    \in \mathcal I}$ of $\Gamma$ in which the \loja inequality holds
  with constants $\kappa_i, \theta_i$.
  
  By compactness, we extract a finite subcover from the $U_i$, from
  which we deduce $\delta > 0$, $\kappa > 0$ and $\theta \in (0,1/2]$
  such that whenever $d(D,\Gamma) < \delta$,
  \begin{align}
    E(D)^{1-\theta} \leq \kappa \norm{\nabla_{\mathcal P}E(D)} =
    \kappa \norm{[D,C_D]} = \kappa \norm{C_D}.
  \end{align}
  (recall from Section \ref{manifold} that $\norm{[D,[D,M]]} = \norm{[D,M]}$
  for $M$ symmetric)

  To apply the \loja inequality to our iteration, it remains to show
  that $d(D_k, \Gamma)$ tends to zero. Suppose this is not the
  case. Then we can extract a subsequence, still denoted by $D_k$,
  such that $d(D_k, \Gamma) \geq \eps$ for some $\eps > 0$. By
  compactness of $\mathcal P$ we extract a subsequence that converges
  to a $D$ such that $d(D, \Gamma) \geq \eps$ and (by continuity)
  $E(D) = 0$, a contradiction. Therefore $d(D_k, \Gamma) \to 0$, and
  for $k$ larger than some $k_0$,
  \begin{align}
    E(D_k)^{1-\theta} &\leq \kappa \norm{C_{k}}. \label{loja_final}
  \end{align}

  For $k \geq k_{0}$,
  \begin{align*}
    E( D_{k})^\theta - E( D_{k+1})^\theta &\geq \frac{\theta}{E( D_{k})^{1-\theta}}
    (E( D_{k}) - E( D_{k+1}))\\
    &\geq \frac \theta {\kappa \norm{C_k}}(E( D_{k}) - E( D_{k+1}))\\
    &\geq \frac {\theta \beta} {\kappa}  \norm{C_{k}}\\
    &\geq \frac{\theta \beta}{\kappa t} \norm{D_{k+1} - D_k}
    + o(\norm{D_{k+1} - D_{k}})
  \end{align*}
  hence
  \begin{align}
    \frac{\theta \beta}{\kappa t} \norm{D_{k+1} - D_k} + o(\norm{D_{k+1} - D_{k}})
    \leq E( D_{k})^\theta - E(
    D_{k+1})^\theta.\label{dk_cauchy}
  \end{align}

  As the right-hand side is summable, so is the left-hand side, which
  implies that $\sum \norm{D_{k+1} - D_k} < \infty$.
  $D_k$ is therefore Cauchy and converges to some limit
  $D_\infty$. $C_k \to 0$, so $D_\infty$ is a critical point.

  Note that now that we know that $D_{k}$ converges to $D_{\infty}$,
  we can replace the $\theta$ and $\kappa$ in this inequality by the
  ones obtained from the \loja inequality around $D_{\infty}$ only.
\end{proof}

Let
\begin{align*}
  e_k = \sum_{l=k}^\infty \norm{D_{l+1} - D_l}.
\end{align*}
This is a (crude) measure of the error at iteration number $k$. In
particular, $\norm{D_k - D_\infty} \leq e_k$.

\begin{theorem}[Convergence rate of the gradient algorithm]
  \
  \begin{enumerate}
  \item If $\theta = 1/2$ (non-degenerate case), then for any $\nu' <
    \frac \beta {2\kappa^{2}}$, there exists $c > 0$ such that
    \begin{align}
      e_k \leq c (1 - \nu')^{k}.
    \end{align}

  \item If $\theta \neq 1/2$ (degenerate case), then there exists $c >
    0$ such that
    \begin{align}
      e_k \leq c k^{-\frac {\theta}{1 - 2\theta}}.
    \end{align}

  \end{enumerate}

\end{theorem}
\begin{proof}
  Summing \eqref{dk_cauchy} from $l = k$ to $\infty$ with $k \geq
  k_{0}$, we obtain
  \begin{align*}
    e_k + o({e_{k}})&\leq\frac{\kappa t}{\theta \beta}E(D_k)^\theta \\
    \lp\frac{\theta \beta}{\kappa t}e_k + o({e_{k}})\rp^{\frac{1-\theta}\theta} &\leq E(D_k)^{1-\theta}\\
    &\leq\kappa \norm{C_k}\\
    &\leq \frac \kappa {t} (e_{k} - e_{k+1}) + o({e_{k} - e_{k+1}})
  \end{align*}
  Hence,
  \begin{align*}
    e_{k+1} &\leq e_k - \nu e_k^{\frac{1-\theta}\theta} + o({e_k}^{\frac{1-\theta}\theta}), \text{ with}\\
    \nu &= \frac{t}{\kappa} \lp\frac{\theta \beta}{\kappa t
    }\rp^{\frac{1-\theta}\theta}
  \end{align*}
  Two cases are to be distinguished. If $\theta = \frac12$, the above
  inequality reduces to
  \begin{align*}
    e_{k+1} &\leq (1 - \nu + o (1)) e_{k}
  \end{align*}
  with $\nu = \frac \beta {2 \kappa^{2}}$ and the result follows.

  The case $\theta \neq 1/2$ is more involved.  We define
  \begin{align*}
    y_k &= c k^{-p},
  \end{align*}
  which verifies
  \begin{align*}
    y_{k+1} &= c (k+1)^{-p}\\
    &= ck^{-p} (1 + 1/k)^{-p}\\
    &\geq ck^{-p} (1 - \frac p k)\\
    &\geq y_k (1 - p c^{-\frac 1 p} y_k^{\frac 1 p})
  \end{align*}

  We set
  $p = \frac {\theta}{1 - 2\theta}$
  and $c$ large enough so that $c > (\frac{\nu}{p})^{-p}$ and
  $y_{k_0} \geq e_{k_0}$. We then prove by induction $e_k \leq y_k$
  for $k \geq k_{0}$. The result follows by increasing $c$ to ensure
  that $e_{k} \leq y_{k}$, for $k < k_{0}$.
\end{proof}

In the non-degenerate case $\theta = 1/2$ (which was the case for the
numerical simulations we performed, see Section \ref{numres}), the
convergence is asymptotically geometric with rate $1 - \nu$, where
\begin{align*}
  \nu &= \frac{\beta}{2 \kappa^2}.
\end{align*}

With the choice $t = \frac 1 \alpha$ suggested by our bounds, the
convergence rate is
\begin{align*}
  \nu &= \frac 1 {4\kappa^2 \alpha}.
\end{align*}

\section{Convergence of the Roothaan algorithm}
The Roothaan algorithm (also called simple SCF in the literature) is
based on the observation that a minimizer $D$ of the energy satisfies
the \emph{Aufbau} principle: $D$ is the projector onto the space
spanned by the eigenvectors associated with the first $N$ eigenvalues
of $F(D)$. This suggests a simple fixed-point algorithm: take for
$D_{k+1}$ the projector onto the space spanned by the eigenvectors
associated with the first $N$ eigenvalues of $F(D_{k})$, and
iterate. Unfortunately, this procedure does not always work: in some
circumstances, oscillations between two states occur, and the
algorithm never converges. This behavior was explained mathematically
in \cite{cances2000}, where the authors notice that the Roothaan
algorithm minimizes the bilinear functional
\begin{align*}
  E(D, D') = \Tr (h(D+D')) + \Tr (G(D) D')
\end{align*}
with respect to its first and second argument alternatively. Thanks to
the \loja inequality, we can improve on their result and prove the
convergence or oscillation of the method.

The bilinear functional verifies $E(D, D') = E(D',D)$, $E(D,D) = 2
E(D)$. In fact, $\frac 1 2 E(\cdot,\cdot)$ is the symmetric bilinear
form associated with the quadratic form $E(\cdot)$. In the following,
we assume the uniform well-posedness hypothesis of \cite{cances2000},
\ie that there is a uniform gap of at least $\gamma >0$ between the
eigenvalues number $N$ and $N+1$ of $F(D_{k})$. Under this assumption,
it can be proven \cite{cances2003computational} that there is a
decrease of the bilinear functional at each iteration
\begin{align*}
  E(D_{k+1},D_{k+2}) &=E(D_{k+2},D_{k+1})\\
  &= \min_{D \in \mathcal P} E (D,D_{k+1})\\
  &\leq E(D_{k},D_{k+1}) - \gamma \norm{D_{k+2} - D_{k}}^{2}
\end{align*}

Since $E(\cdot,\cdot)$ is bounded from below on $\mathcal P \times
\mathcal P$, this immediately shows that $D_{k} - D_{k+2} \to 0$,
which shows that $D_{2k}$ and $D_{2k+1}$ converge up to extraction,
which was noted in \cite{cances2000}. We now prove convergence of
these two subsequences, again using the \loja inequality.

$E(\cdot,\cdot)$ is defined on $\mathcal P \times \mathcal P$, which
inherits the Riemannian structure of $\mathcal P$ by the natural inner
product $\lela (D_{1},D_{1}'), (D_{2},D_{2}')\rira = \lela D_{1},
D_{2}\rira + \lela D_{1}',D_{2}'\rira$. In this setting, the gradient is
\begin{align*}
  \nabla_{\mathcal P \times \mathcal P} E(D,D') &= \mat{{[D, F(D')]}\\{[D',F(D)]}}.
\end{align*}

and therefore, using the fact that $D_{k+1}$ (resp. $D_{k+2}$) and
$F(D_{k})$ (resp $F(D_{k+1})$) commute,
\begin{align*}
  \norm{\nabla_{\mathcal P \times \mathcal P} E(D_{k},D_{k+1})} &=
  \sqrt{\norm{[D_{k}, F(D_{k+1})]}^{2} + \norm{[D_{k+1}, F(D_{k})]}^{2}}\\
  &=\norm{[D_{k}, F(D_{k+1})]}\\
  &=\norm{[D_{k} - D_{k+2}, F(D_{k+1})]}\\
  &\leq 2 \norm{F(D_{k+1})}_{\text{op}} \norm{D_{k+2} - D_{k}}
\end{align*}

A trivial extension of Theorem \ref{lojap} to the case of a
functional defined on $\mathcal P \times \mathcal P$ shows that we
can apply the \loja inequality to $E(\cdot,\cdot)$. By the same
compactness argument as before, the inequality
\begin{align*}
  E(D_{k}, D_{k+1})^{1 - \theta'} &\leq \kappa'   \norm{\nabla_{\mathcal P \times \mathcal P} E(D_{k},D_{k+1})}\\
  &\leq 2 \kappa' \norm{F(D_{k+1})}_{\text{op}} \norm{D_{k+2} - D_{k}}
\end{align*}
holds for $k$ large enough, with constants $\kappa' > 0$ and $\theta'
\in (0,\frac 1 2]$.

The exact same reasoning as for the gradient algorithm proves the
following theorems
\begin{theorem}[Convergence/oscillation of the Roothaan algorithm]
  \label{theo-roothaan}
  Let $D_{0} \in \mathcal P$ such that the sequence $D_{k}$ of
  iterates generated by the Roothaan algorithms verifies the uniform
  well-posedness hypothesis with uniform gap $\gamma > 0$. Then the
  two subsequences $D_{2k}$ and $D_{2k+1}$ are convergent. If both
  have the same limit, then this limit is a solution of the
  Hartree-Fock equations.
\end{theorem}

\begin{theorem}[Convergence rate of the Roothaan algorithm]
  Let $D_{k}$ be the sequence of iterates generated by a uniformly
  well-posed Roothaan algorithm, and let
  \begin{align*}
    e_k = \sum_{l=k}^\infty \norm{D_{l+2} - D_l}.
  \end{align*}
  Then,
  \begin{enumerate}
  \item If $\theta' = 1/2$ (non-degenerate case), then for any $\nu' <
    \frac{\gamma}{8 {\kappa'}^{2} \norm{F}_{\text{op}}^{2}}$, where
    $\norm{F}_{\text{op}}$ is the uniform bound on $F$ proved in
    \eqref{bound_F}, there exists $c > 0$
    such that
    \begin{align}
      e_k \leq c (1 - \nu')^{k}.
    \end{align}

  \item If $\theta' \neq 1/2$ (degenerate case), then there exists $c >
    0$ such that
    \begin{align}
      e_k \leq c k^{-\frac {\theta'}{1 - 2\theta'}}.
    \end{align}

  \end{enumerate}
\end{theorem}




\section{Level-shifting}
\label{ls}
The Level-Shifting algorithm was introduced in
\cite{saunders1973level} as a way to avoid oscillation in
self-consistent iterations. By shifting the $N$ lowest energy levels
(eigenvalues of $F$), one can force convergence, although denaturing
the equations in the process. We now prove the convergence of this
algorithm.

The same arguments as before apply to the functional
\begin{align*}
  E^{b}(D,D') &= \Tr (h(D+D')) + \Tr(G(D)D') + \frac b 2
  \norm{D-D'}^{2}\\
  &= \Tr (h(D+D')) + \Tr(G(D)D') - b\Tr(DD') + b N\\
\end{align*}
with associated Fock matrix $F^{b}(D) = F(D) - bD$. The difference
with the Roothaan algorithm is that for $b$ large enough, there is a
uniform gap $\gamma^{b} > 0$, and $D_{k} - D_{k+1}$ converges to 0
\cite{cances2000}. Therefore, we have the following theorems

\begin{theorem}[Convergence of the Level-Shifting algorithm]
  Let $D_{0} \in \mathcal P$. Then there exists $b_{0} > 0$ such
  that for every $b > b_{0}$, the sequence $D_{k}$ of iterates
  generated by the Level-Shifting algorithm with shift parameter b
  verifies the uniform well-posedness hypothesis with uniform gap
  $\gamma > 0$ and converges.
\end{theorem}

\begin{theorem}[Convergence rate of the Level-Shifting algorithm]
  Let $D_{k}$ be the sequence of iterates generated by the
  Level-Shifting with shift parameter $b > b_{0}$, and let
  \begin{align*}
    e_k = \sum_{l=k}^\infty \norm{D_{l+2} - D_l}.
  \end{align*}
  Then,
  \begin{enumerate}
  \item If $\theta' = 1/2$ (non-degenerate case), then for any $\nu' <
    \frac{\gamma^{b}}{8 {\kappa'}^{2} \norm{F^{b}}_{\text{op}}^{2}}$,
    there exists $c > 0$
    such that
    \begin{align}
      e_k \leq c (1 - \nu')^{k}.
    \end{align}

  \item If $\theta' \neq 1/2$ (degenerate case), then there exists $c >
    0$ such that
    \begin{align}
      e_k \leq c k^{-\frac {\theta'}{1 - 2\theta'}}.
    \end{align}

  \end{enumerate}
\end{theorem}

We can use this result to heuristically predict the behavior of the
algorithm when $b$ is large. $\gamma^{b}$ and
$\norm{F^{b}}_{\text{op}}$ both scale as $b$ for large values of
$b$. Assuming non-degeneracy, we can take $\kappa' > \frac 1 {\sqrt
  {2\abs{\lambda_{1}}}}$, where $\lambda_{1}$ is the eigenvalue of
smallest magnitude of the Hessian $H_{1} + \frac b 2 H_{2}$, where
$H_{1} = H_{\mathcal P \times \mathcal P} E (D^{\infty},D^{\infty})$
and $H_{2} = H_{\mathcal P \times \mathcal P} \norm{D - D'}^{2}
(D^{\infty},D^{\infty})$. But $H_{2}$ admits zero as an eigenvalue
(for instance, note that $\norm{D-D^{'}}^{2}$ is constant along the
curve $(D_{t},D_{t}') = (U_{t} D U_{t}^{T},U_{t} D' U_{t}^{T})$, where
$U_{t}$ is a family of orthogonal matrices), so that, when $b$ goes to
infinity, $\lambda_{1}$ tends to the eigenvalue of smallest magnitude
of $H_1$ restricted to the nullspace of $H_2$, and therefore stays
bounded. Therefore, $\nu'$ scales as $\frac 1 b$, which suggests that $b$
should not be too large for the algorithm to converge quickly.
\section{Numerical results}
\label{numres}

We illustrate our results on atomic systems, using gaussian basis
functions. The gradient method was implemented using the software
Expokit \cite{sidje1998expokit} to compute matrix exponentials. In our
computations, the cost of a gradient step is not much higher than a
step of the Roothaan algorithm, since the limiting step is computing
the Fock matrix, not the exponential. However, the situation may
change if the Fock matrix is computed using linear scaling
techniques. In this case, one can use more efficient ways of computing
geodesics, as described in \cite{edelman1998geometry}.

First, the \loja inequality with exponent $\frac 1 2$ was checked to
hold in the molecular systems and basis sets we encountered,
suggesting that the minimizers are non-degenerate. Consequently, we
never encountered sublinear convergence of any algorithm.

For a given molecular system and basis, we checked that the
Level-Shifting algorithm converged as $(1-\nu)^{k}$, where $\nu$ is
asymptotically proportional to $\frac 1 b$, which we predicted
theoretically in Section \ref{ls} (see \figref{fig:ls}). This means that the estimates we
used have at least the correct scaling behavior.

\begin{figure}[H]
  \centering
  \scalebox{0.5}{ \includegraphics{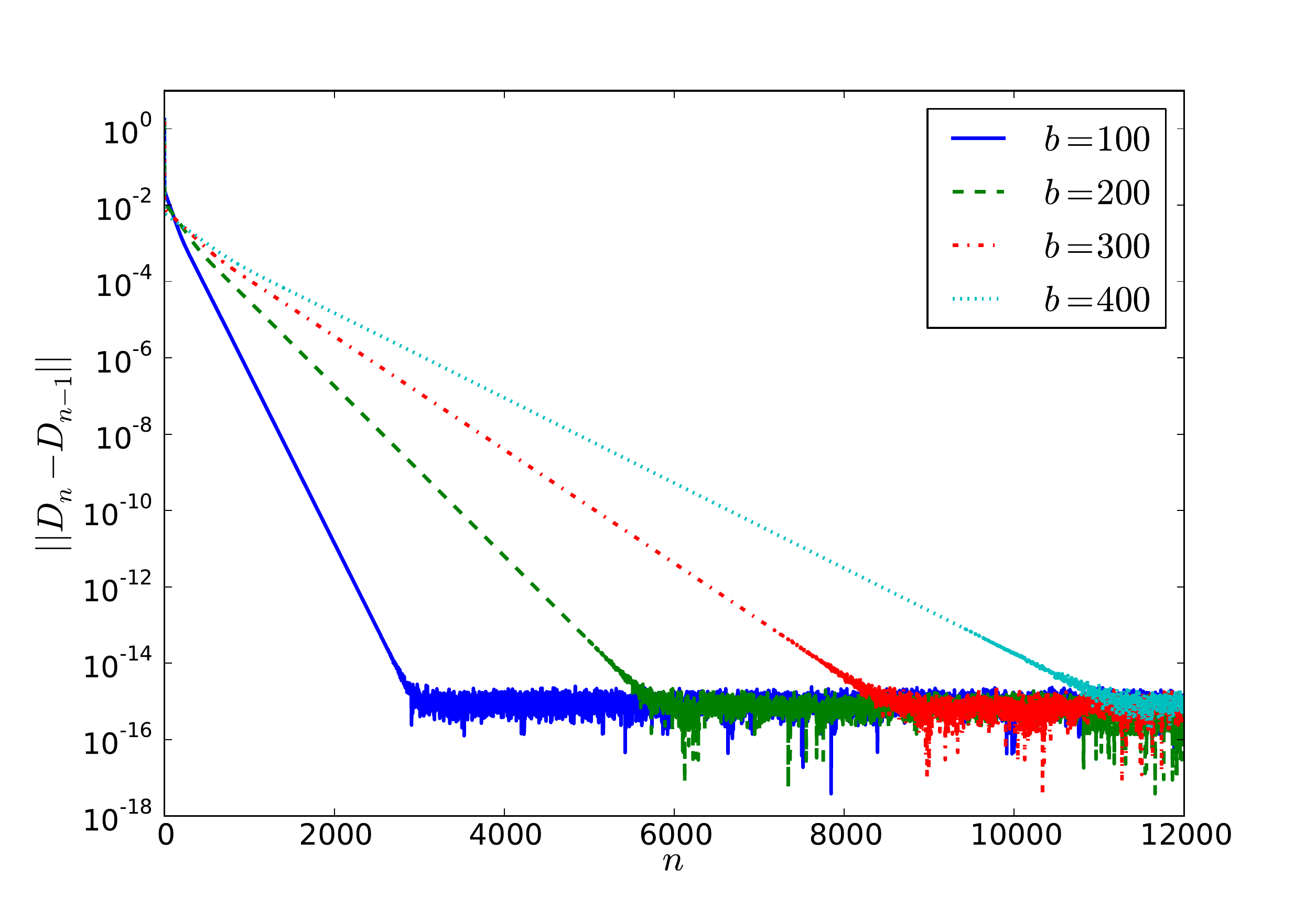}}
  \caption{Convergence of the Level-Shifting algorithm. The
    convergence is linear until machine precision. The horizontal
    spacing of the curves reveals the asymptotic relationship $\nu
    \propto \frac 1 b$. The system considered is the carbon atom
    ($N=Z=6$), under the RHF formalism, using the 3-21G gaussian basis
    functions.}
  \label{fig:ls}
\end{figure}

Next, we compared the efficiency of the Roothaan algorithm and of the
gradient algorithm, in the case where the Roothaan algorithm
converges. Our analysis leads to the estimate $\nu = \frac {\gamma}{8
  {\kappa'}^{2} \norm{F}_{\text{op}}^{2}}$ for the Roothaan algorithm,
and $\nu = \frac 1 {4\kappa^2 \alpha}$ for the gradient algorithm with
stepsize $t = \frac 1 \alpha$.

It is immediate to see that, up to a constant multiplicative factor,
$\kappa' > \kappa$, $\gamma \leq \norm{F}_{\text{op}}$ and for the
cases of interest $\alpha \approx \norm{F}_{\text{op}}$, so from our
estimates we would expect the gradient algorithm to be faster than the
Roothaan algorithm. However, in our tests the Roothaan algorithm was
considerably faster than the gradient algorithm (see
\figref{fig:roothaan-grad}). This conclusion holds even when the
stepsize is adjusted at each iteration with a line search.

\begin{figure}[h!]
  \centering
  \scalebox{0.5}{ \includegraphics{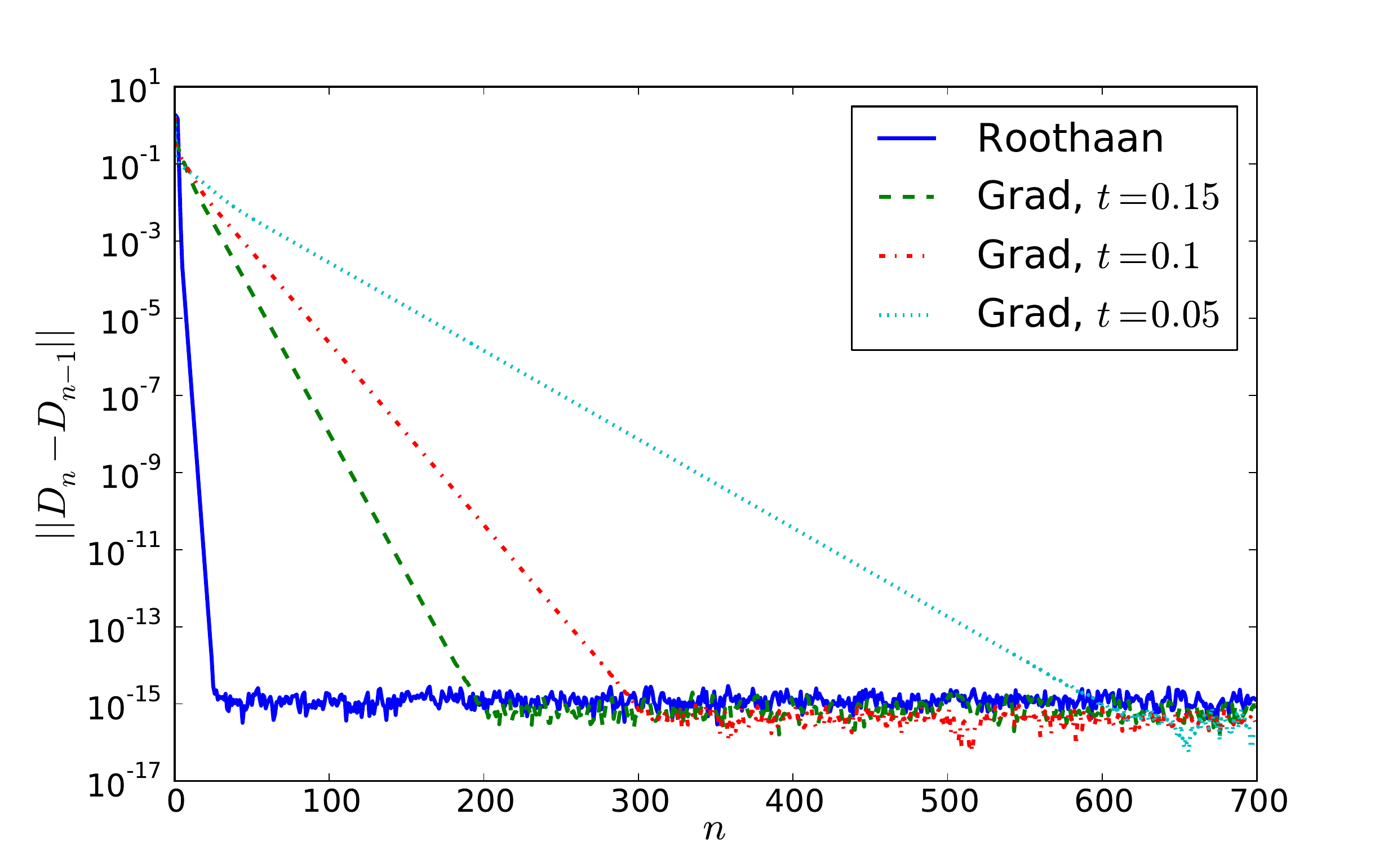}}
  \caption{Comparison of Roothaan and gradient algorithm. The system
    considered is the carbon atom ($N=Z=6$), under the RHF formalism,
    using the 3-21G gaussian basis functions.}
  \label{fig:roothaan-grad}
\end{figure}

The reason that the Roothaan algorithm performs better than expected
is that the inequality
$$\norm{[D_{k+2} - D_{k},F(D_{k+1})]} \leq 2 \norm{F(D_{k+1})}
\norm{D_{k+2} - D_{k}}$$ is very far from optimal. Whether an improved
bound (in particular, one that does not depend on the dimension) can
be derived is an interesting open question.

The outcome of these tests seems to be that the gradient algorithm is
slower. It might prove to be faster in situations where the gap is
small, or whenever $\kappa'$ is much larger than $\kappa$. We have
been unable to find concrete examples of such cases.

\section{Conclusion, perspectives}
In this paper, we introduced an algorithm based on the idea of
gradient descent. By using the analyticity of the objective function
and of the constraint manifold, we were able to derive a \loja
inequality, and use that to prove the convergence of the gradient
method, under the assumption of a small enough stepsize. Next,
expanding on the analysis of \cite{cances2000}, we extended the \loja
inequality to a Lyapunov function for the Roothaan algorithm. By
linking the gradient of this Lyapunov function to the difference in
the iterates of the algorithm, we proved convergence (or oscillation),
an improvement over previous results which only prove a weaker version
of this. In this framework, the Level-Shifting algorithm can be seen
as a simple modification of the Roothaan algorithm, and as such can be
treated by the same methods. In each case, we were also able to derive
explicit bounds on the convergence rates.

The strength of the \loja inequality is that no higher-order
hypothesis are needed for its use. As a consequence, the rates of
convergence we obtain weaken considerably if the algorithm converges
to a degenerate critical point. A more precise study of the local
structure of critical points is necessary to understand why the
algorithms usually exhibit geometric convergence. This is related to
the problem of local uniqueness and is likely to be hard (and, indeed,
to our knowledge has not been tackled yet).

Even though our results hide the complexity of the local structure in
the constants of the \loja inequality, they still provide insight as
to the influence of the basis on the speed of convergence, and can be
used to compare algorithms. All of our results use crucially the
hypothesis of a finite-dimensional Galerkin space. For the gradient
algorithm, we need it to ensure the existence of a stepsize that
decreases the energy. This is analogous to a CFL condition for the
discretization of the equation $\od D t = -[D,[D,F_{D}]]$, and can
only be lifted with a more implicit discretization of this
equation. For the Roothaan and Level-Shifting algorithms, we use the
finite dimension hypothesis to bound $F(D)$. As noted in Section
\ref{numres}, the inequality is not sharp, so it could be that the
infinite-dimensional version of the Roothaan and Level-Shifting
algorithms still converge. More research is needed to examine this.

The gradient algorithm we examined only converges towards a stationary
point of the energy, that may not be a local minimizer, or even an
Aufbau solution. However, it will generically converge towards a local
minimizer, unlike the Level-Shifting algorithm with large
$b$. Therefore, it is the most robust of the algorithms
considered. Although it was found to be slower than other algorithms
on the numerical tests we performed, it has the advantage that its
convergence rate does not depend on the gap $\lambda_{N+1} -
\lambda_{N}$, and might therefore prove useful in extreme situations.

An algorithm that could achieve the speed of the fixed-point
algorithms with the robustness granted by the energy monotonicity
seems to be the ODA algorithm of Cancès and Le Bris\cite{cances2000},
along with variants such as EDIIS, or combinations of EDIIS and DIIS
algorithms\cite{kudin2002black}. We were not able to examine these
algorithms in this paper. At first glance, the ODA algorithm should
fit into our framework (indeed, the ODA algorithm was built to satisfy
an energy decrease inequality similar to
\eqref{energy_decrease}). However, it works in a relaxed parameter
space $\tilde{P}$, and using the commutator to control the differences
of iterates as we did only makes sense on $\mathcal P$. Therefore,
other arguments have to be used.

A variant on the gradient algorithm used here is to modify the local
geometry of the manifold $\mathcal P$ by using a different inner
product, leading to a variety of methods, including conjugate gradient
algorithms\cite{edelman1998geometry}. These methods fit into our
framework, as long as one can prove that they are ``gradient-like'',
in the sense that one can control the gradient by the difference
$D_{k+1} - D_{k}$. However, precise estimates of convergence rates
might be hard to obtain.

Also missing from the present contribution is the study of other
commonly used algorithms, such as DIIS\cite{pulay1982improved}, and
variants of (quasi)-Newton
algorithms\cite{bacskay1981,host2008augmented}. DIIS numerically
exhibits a complicated behavior that is probably hard to explain
analytically, and the (quasi)-Newton algorithms require a study of the
second-order structure of the critical points, which we are unable to
do.

\section{Acknowledgments}
The author would like to thank Eric Séré for his extensive help,
Guillaume Legendre for the code used in the numerical simulations and
Julien Salomon for introducing him to the \loja inequality. He also
thanks the anonymous referees for many constructive remarks.

\bibliographystyle{plain}
\bibliography{refs.bib}

\end{document}